\documentclass[review]{elsarticle}

\usepackage{hyperref}

\usepackage{amsmath}
\usepackage{amssymb}
\usepackage{amsthm}
\usepackage{amsxtra}
\usepackage{amsfonts}
\usepackage{graphicx}
\usepackage{latexsym}
\usepackage{epsfig}
\usepackage{verbatim}
\usepackage{bm}
\usepackage{bbm}
\usepackage{color}
\usepackage{varwidth}
\usepackage{subfigure}
\usepackage{wrapfig}
\usepackage{trivfloat}
\usepackage{todonotes}
\usepackage{url}
\usepackage{fullpage}
\usepackage{dsfont}
\usepackage[margin=9pt,font=small,labelfont=bf,labelsep=colon]{caption}
\usepackage{natbib}

\usepackage{lipsum}
\usepackage[procnames]{listings}
\usepackage{array}
\usepackage{multicol}
\usepackage{float}
\usepackage{url}
\usepackage{multirow}
\usepackage{enumitem}
\usepackage{amsmath,lipsum}
\usepackage[doublespacing]{setspace}
\usepackage[subfigure]{tocloft}

\usepackage[subtle]{savetrees}

\newtheorem{theorem}{Theorem}
\newtheorem{lemma}{Lemma}
\newtheorem{proposition}{Proposition}

\newtheorem{definition}{Definition}
\newtheorem{corollary}{Corollary}[theorem]

\graphicspath{ {./figures/} }

\journal{XX - under review}
\numberwithin{equation}{section}

\begin{document}

\begin{frontmatter}

%\title{Link Dimension for Exact Construction of Graphs}
\title{Link Dimension and Exact Construction of a Graph}

\author{Gunjan S. Mahindre and Anura P. Jayasumana}%\footnote[1]{Corresponding Author: Anura Jayasumana, anura.jayasumana{@colostate.edu}}}
\address{Department of Electrical and Computer Engineering, Colorado State University, Fort Collins, Colorado 80525}
\ead{gunjan.mahindre,anura.jayasumana{@colostate.edu} }
\cortext[mycorrespondingauthor]{Corresponding Author: Anura Jayasumana, anura.jayasumana{@colostate.edu}}
\begin{abstract}
Minimum resolution set and  associated metric dimension  provide the basis for unique and systematic labeling of nodes of a graph using %their
distances to a set of landmarks. %However, the corresponding set of  distance vectors is  not unique to the original graph and does not facilitate its exact reconstruction. 
Such a distance vector set, however, may not be unique to the graph and does not allow for its exact construction. 
%Such a distance vector set, however,  may not facilitate the exact construction of the original graph.
The concept of construction set is presented, which facilitates the unique representation of nodes and the graph as well as its exact construction. %using distance vectors to landmarks. 
Link dimension %of a graph % unlike metric dimension, 
is the minimum number of landmarks in a construction set. Results presented include  necessary conditions for a set of landmarks to be a construction set,  bounds for link dimension, and guidelines for transforming a resolution set to a construction set.
%The concepts of construction set and link dimension of a graph are introduced. Link dimension % unlike metric dimension, 
%specifies the minimum number of landmarks that allows the exact construction of a graph using the hop distances from the landmarks to each node. We  present a number of properties including the necessary conditions for a resolution set to be a valid construction set and bounds for link dimension. 
\end{abstract}

\begin{keyword}
%\texttt{metric dimension, graph reconstruction, resolution set, links, landmarks, anchors, bounds}
\texttt{metric dimension, graph construction, resolution set, network coordinates,\\ virtual coordinates}
% \MSC[2018] 00-01\sep  99-00
\end{keyword}

\end{frontmatter}

%\linenumbers

\section{Introduction}
\label{interoduction}
\setlength{\belowdisplayskip}{0pt} \setlength{\belowdisplayshortskip}{0pt}
\setlength{\abovedisplayskip}{0pt} \setlength{\abovedisplayshortskip}{0pt}
Consider a simple  undirected connected graph  $G$, defined by
$G=\left\{\mathcal{V}, \mathcal{E}\right\}$, where $\mathcal{V}$ is the set of nodes of cardinality $ N $ and $\mathcal{E}$ is the set of edges (links). 
$ G $ may be represented in terms of its adjacency matrix $ \textbf{A} $, % the $ {ij}^{th} $ element of which is given by $ a_{ij} $, 
where
%\begin{equation}
$\textbf{A} = \left[a_{ij}| a_{ij}=1 \;\text{ if }\; (i,j) \in \mathcal{E},
\;\text{ 0 otherwise}\; \right].
\label{A definition}
$%\end{equation}
Nodes $ i $ and $ j $ are said to be adjacent if $ a_{ij} = 1 $. Alternatively, $ G $ can be represented by its distance matrix $\textbf{H}.
%Nodes $ i $ and $ j $ are said to be adjacent if there is an edge between them. $ G $ can be represented by its distance matrix $\textbf{H}
\in \mathbb{Z}^{N \times N}$, where the $ {ij}^{th} $ element is given by $ h_{ij} $, such that
\begin{equation}
	\begin{aligned}
	\textbf{H} = \left[h_{ij}| h_{ij}=\text{the number of links in the shortest path (path length) } \text{from node} \ i \ \text{to node} \ j \right].
%		\textbf{H} = \left[h_{ij}| h_{ij}=\text{the length of the shortest path } \text{from node} \ i \ \text{to node} \ j \right].
		%h_{ij} = & \ \text{the length of the shortest path } \text{from node} \ i \ \text{to node} \ j .
	\end{aligned}
	\label{H definition}
\end{equation}
\noindent
%If we assume that a node in a graph can sense distances from a set of landmark nodes, $L$, with cardinality $|L|$, then a complete set of distances to each node from all the landmarks is called a \textit{distance vector matrix} $ \textbf{P} $ \cite{1}, of size $(N \times |L|)$ which is a subset of the complete distance matrix  $ \textbf{H} $, of size $(N \times N)$. Note that in networking context, graphs are termed as networks, landmarks as anchors, and distance vectors as virtual coordinates \cite{2}. 
% Let's denote the $ i^{th} $ row of $ \textbf{H} $ by a distance vector $ \textbf{H}(i)=\langle h_{i1},h_{i2},..,h_{iN}\rangle $ where $ i \in \mathcal{V} $.
% and $ h_{ij} $ represents shortest hop distance from node $ i $ to node $ j $.
%then $ \textbf{P}_ $ represents the distance vector of node $ i $ to all other $ i \in \textit{V} $. 
Let $ \mathcal{M}=\{A_1,...A_m\}$ be a subset of nodes of $ \mathcal{V}, $ with cardinality $m = |\mathcal{M}|$, designated as landmarks. The set of distances to each landmark from all $ i \in \mathcal{V} $ forms % is called 
a \textit{distance vector matrix} $ \textbf{P}_\mathcal{M} $ \cite{1} of size $(N \times m)$, where the $ i^{th} $ row of $ \textbf{P}_\mathcal{M} $ is  the distance vector (DV)
%$ \textbf{P}_\mathcal{M}(i)=\langle h_{iA_1},h_{iA_2},..,h_{iA_m}\rangle $; $ A_k \in \mathcal{M} $ and $ k \in \{1,2,..,m\} $. 
\begin{equation}
\textbf{P}_\mathcal{M}(i)=\langle h_{iA_1},h_{iA_2},..,h_{iA_m}\rangle ;\ \   i=\{1,2,..,N\}.  
\end{equation} 
Note that $h_{iA_k} = 0$ for $i = A_k$ and $h_{iA_k} >0$ for $i \neq A_k$. $ \textbf{P}_\mathcal{M} $ consists of a subset of columns of $\textbf{H}$.
% Without loss of generality, we assume these landmarks to be nodes $ 1,2,..,m $ and thus, $ \textbf{P}_\mathcal{M} $ is the $(N \times m)$ matrix corresponding to first $ m $ columns of $ \textbf{P} $.

%which is a subset of the complete distance matrix  $ \textbf{H} $, of size $(N \times N)$.

Distance vector based methods are attractive for many communication and social networking operations and applications. In the context of networking, graphs are termed as networks, landmarks as anchors, and distance vectors to  landmarks as \textit{virtual coordinates} \cite{2,10,Pendharkar2019}. Virtual coordinate based techniques are used to overcome, for example, uncertainties of physical distance measurements required for Cartesian coordinates caused by fading or interference of radio signals in wireless sensor networks \cite{2} and inaccessibility of certain  nodes in social networks \cite{icc}.  Conceptually, the adjacency matrix (Eq. \ref{A definition}) and the distance matrix (Eq.\ref{H definition}) are equivalent in representing a graph as one can be derived from the other. However, as  explained in \cite{Beerliova2006} using examples from communication
networks, it is often realistic to obtain the distances between nodes (i.e., $h_{ij}$'s) 
in many communication networks, while it is difficult or impossible to obtain
information about the presence  or absence of specific edges (i.e., $a_{ij}$'s) that are far away
from the query node.   Landmarks are typically chosen randomly or based on heuristics \cite{Dhanapala2011h}, only to ensure the uniqueness of coordinates of nodes without regard to their suitability of the DVs for reconstruction of the graph or capture the topology. Topology or layout information is recovered from  these distance vectors using approaches such as low-rank matrix completion \cite{JayArx2018} or by exploiting the statistical characteristics observed in the class of networks \cite{ding_paper}. There is no formal basis to relate the landmark selection and the corresponding  distance vectors to the ability to recover the network topology, a problem that we address in this paper.  
%As the selection of optimal set of landmarks is computation intensive \cite{1}, it is customary to over-anchor the network (i.e., select a significant fraction of nodes as landmarks) to ensure uniqueness virtual coordinates \cite{icc}.   

%While there are approaches for representing graphs in the form of unique node coordinates 
%%which are unique for each node in $ G $, 
%\cite{6}, they do not guarantee the
%%complete and unique 
%recovery of exact $ G $ from its representation. 
%By the definition of \textit{metric dimension}, 
With a sufficient number of landmarks, the rows of $ \textbf{P}_\mathcal{M} $ can uniquely label each node in $G$. A set of such landmarks is called a \textit{resolution set}, $ \mathcal{R} $,  and the minimum possible cardinality of such a set is called the \textit{metric dimension}, $\beta(G)$, of $G$. Let $ \tilde{\mathcal{R}} $ be  such a minimum resolution set; thus $\beta(G)=|\tilde{\mathcal{R}}|$. The concept of metric dimension was presented in \cite{8,3} and \cite{1,4,5,6,7} among others extended the results for different families of graphs. 
%While these approaches represent graphs in the form of unique node coordinates 
%%which are unique for each node in $ G $,  \cite{6}, 
%they do not guarantee the
%%complete and unique 
%recovery of exact $ G $ from the distance vector representation.
Although these approaches allow the representation of each node of G with unique virtual  coordinates,  in general there are multiple graphs that satisfy the same set of coordinates, and therefore the construction of the original graph from the distance vectors is not possible. 

We address the construction of a graph from distance vectors to a small set of landmarks,   an especially  important concept for large-scale networks as it dramatically reduces the complexity of network measurement, topology extraction and other network analytics \cite{icc} \cite{ding_paper}.  After identifying the conditions that make the minimum resolution set not sufficient for graph reconstruction (Section 2),  we present the concepts of link dimension and construction set for exact construction of the original graph using a set of minimum length distance vectors (Section 3). The DVs of a construction set provide a complete and exact representation of a graph in addition to assigning unique coordinates to its nodes. Several related  properties and bounds are derived (Section 4). 

\section{Metric dimension and graph reconstruction}
\label{metric dimension and graph reconstruction}
\setlength{\belowdisplayskip}{0pt} \setlength{\belowdisplayshortskip}{0pt}
\setlength{\abovedisplayskip}{0pt} \setlength{\abovedisplayshortskip}{0pt}
%Set of distance vectors, $ \textbf{P}_{\tilde{\mathcal{R}}} $, corresponding to the minimum resolution set $ \tilde{\mathcal{R}} $ of a graph uniquely labels each node.
Let $ \textbf{P}_{\tilde{\mathcal{R}}}  $ be the $ (N \times \beta(G)) $ matrix of unique distance vectors for $ \tilde{\mathcal{R}}. $  %, where $ m=|\tilde{\mathcal{R}}| $. 
Though the distances from nodes in $ \tilde{\mathcal{R}} $ to other nodes allow unique representation of each node, we show that $ \textbf{P}_{\tilde{\mathcal{R}}} $ does not guarantee the ability to reconstruct the original graph which a complete $ \textbf{H} $ or $ \textbf{A} $ does.

\begin{theorem}
	The set of unique distance vectors, $ \textbf{P}_{\tilde{\mathcal{R}}} $, corresponding to $ \tilde{\mathcal{R}} $ of graph $ G $, does not guarantee a one-to-one relationship with $ G $ and thus the exact reconstruction of $ G $. 
\end{theorem}
\begin{proof}
	Consider the two graphs shown in Figures \ref{fig_1:figa} and \ref{fig_1:two figures}. Note that each graph is compatible with the unique distance vector set $  \textbf{P}_{\tilde{\mathcal{R}}} = \{\langle0,1\rangle,\langle1,0\rangle,\langle2,1\rangle,\langle2,2\rangle,\langle1,2\rangle\}  $ corresponding to the minimum resolution set of nodes $ \tilde{\mathcal{R}}=\{1,2\} $. 
%	Consider a circular graph $C_n$, with 5 nodes. As $\beta(C_n )=2$; $n \in \textit{odd integers}$ [1], the distance vectors are unique, and metric dimension is satisfied, yet, 
	The presence or absence of edge (5,3) does not change $ \textbf{P}_{\tilde{\mathcal{R}}} $ and conversely, 
%	given $ \textbf{P}_{\tilde{\mathcal{R}}} $
	the presence or absence  of the edges cannot be confirmed from the distance vectors in $ \textbf{P}_{\tilde{\mathcal{R}}} $. % Nodes 5 and 3 are single hop apart, however; both the graphs (with and without edge (5,3)) satisfy $ \textbf{P}_{\mathcal{R}} $.	
\end{proof}
Having a unique graph $ G $ corresponding to $ \textbf{P}_{\mathcal{R}} $ of a resolution set ${\mathcal{R}} $ is desirable
as $ \textbf{P}_{\mathcal{R}} $ then provides an alternative compact representation of $G$. However, there is no formal foundation for the number of distance measurements needed to construct a graph for problems such as topology extraction \cite{icc,ding_paper} . 
%Constructing networks from a small set of landmarks has received significant interest recently for large-scale networks as it dramatically reduces the complexity of network measurement, topology extraction and other network analytics \cite{icc} \cite{ding_paper}. 
The difficulty arises due to edges such as (5,3).
Next, we define terms \textit{invisibility} and \textit{ambiguity} of such edges from two perspectives, that of selecting landmarks (i.e., $ {\mathcal{R}}$) for a graph,  and of constructing $ G $ from $ \textbf{P}_{\mathcal{R}} $ respectively, and illustrate them in Figure 1.
%We define terms \textit{invisibility} and \textit{ambiguity} of such edges from the perspectives of generating $ \textbf{P}_{\mathcal{R}} $ from $ G $ and extracting $ G $ from $ \textbf{P}_{\mathcal{R}} $ respectively.
%keep the following in thesis
%For a graph $G = (\mathcal{V},\mathcal{E})$, let’s consider that $\textbf{A}$ being the original adjacency matrix for $G$, $\hat{\textbf{A}}$ is the predicted (reconstructed) adjacency matrix from distance vector matrix $ \textbf{P}_\mathcal{M} $ created by $\textbf{A}$ and $ \mathcal{M} $. Thus, if one $ \textbf{P}_\mathcal{M} $ generates multiple $ G $ then $ \beta(G) $ is not sufficient to represent $ G $ uniquely and completely. Subsequently, in any given graph $ G $, there may exist few edges $ (i,j) $, while $\textbf{A}[i,j]=1$ ; s.t. $ \textbf{P}_\mathcal{M} $ is not affected by substituting $\hat{\textbf{A}} [i,j]=1$ or 0. Thus, we define the terms \textit{ambiguity} and \textit{invisible edge} as follows:
%\noindent The given $ \textbf{P} $, with $|L| \leq \beta(G)$, does not guarantee sufficient information to complete $ \hat{A} $.
%\textcolor{red}{OR - provide certainty of presence or absence of every edge in $ G $.} 
%\textcolor{red}{\noindent \textit{Ambiguity} in reconstructing $ \hat{A} $ from $ \textbf{P} $ causes multiple reconstruction results for the given distance vector matrix.} 

\begin{definition} \label{def I}
	Invisible edge: An edge (i,j) in a graph is said to be invisible with respect to a resolution set $\mathcal{R}$ when the removal or the addition of edge  $(i,j) $ from $ G $ does not affect any value in $ \textbf{P}_{\mathcal{R}} $.
\end{definition}
\begin{definition} \label{def A}
	Ambiguous edge: An edge is  ambiguous with respect to a distance vector set $ \textbf{P}_{\mathcal{R}}$,   if $ \exists $ two graphs $ G_1 $ and $ G_2 $ that satisfy $ \textbf{P}_{\mathcal{R}} $, where $ (i,j) $ is present in one and absent in the other.
%	entry $ \hat{A}[i,j] $ is termed as ambiguous with respect to $ \mathcal{R} $ only when for a given $ \textbf{P}_\mathcal{R} $, both the graphs corresponding to $ \hat{A}[i,j]=1 $ and 0, satisfy $ \textbf{P}_\mathcal{R} $.
\end{definition}
\begin{figure}[ht] \label{fig_1}
	\centering
%	\begin{subfigure} []
\subfigure[]{
		\centering
		\includegraphics[width=0.3\textwidth]{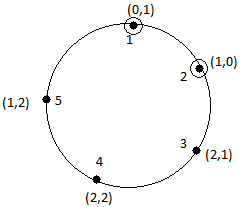}
		\label{fig_1:figa}
	}
%	\end{subfigure} 
%    \begin{subfigure} []
\subfigure[]{
	    \centering
		\includegraphics[width=0.3\textwidth]{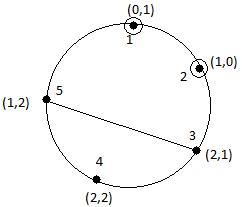}
		\label{fig_1:two figures}
	} 
%\end{subfigure} 
\\
\subfigure[]
	{\centering
		\includegraphics[width=0.32\textwidth]{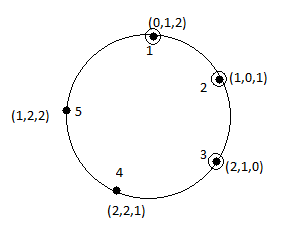}
		\label{fig_1:sdim}
	} 
\subfigure[]{
\centering
		\includegraphics[width=0.3\textwidth]{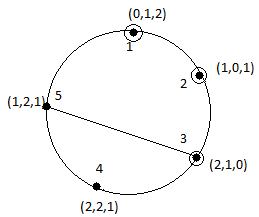}
		\label{fig_1:another graph}
	}
%\end{subfigure} 
	\caption{Ambiguous and invisible edges: Graphs in (a) and (b) both have ${\tilde{\mathcal{R}}}=\{1,2\}$ and   $\textbf{P}_{\tilde{\mathcal{R}}} = \{\langle0,1\rangle,\langle1,0\rangle, \langle2,1\rangle, \\ \langle2,2\rangle,\langle1,2\rangle\}$. Thus edge $(3,5)$ is invisible w.r.t. ${\tilde{\mathcal{R}}}$, and  w.r.t. $  \textbf{P}_{\tilde{\mathcal{R}}}$  edge $(3,5)$ is ambiguous. Distance vectors of the two graphs are shown in (c) and (d) respectively, when node 3 is added as a landmark, i.e., $ \mathcal{C} = \{1,2,3\}$. Edge $(3,5)$ is no longer invisible  w.r.t. $ \mathcal{C}$, and given corresponding $\textbf{P}_{{\mathcal{C}}}$ each graph can be exactly constructed. 	
	}
\label{fig:fig_1}
\end{figure}
\begin{lemma}
	No edge connected to a landmark node is invisible (or ambiguous). 
	\label{lemma_anchor_edge}
	\end{lemma} 
\begin{proof}
	Landmark $k$ is the only node with the corresponding coordinate equal to zero, i.e., $h_{iA_k}=0 \  \textit{if and only if}\  i=A_k$. Furthermore, $h_{iA_k}=1$ if and only if there is an edge between $i$  and $A_k$. Thus  no edge connected to an anchor is invisible.   
\end{proof}
\noindent
%With respect to the difference between distance vectors for nodes $i,j$, let 
To characterize the difference between distance vectors $\textbf{P}_\mathcal{R}(i) \textit{ and } \textbf{P}_\mathcal{R}(j)$ for a pair of nodes $i,j$, we define
\begin{equation}
\begin{split}
\bigtriangleup_{ij} =  {Max \{ \bigtriangleup_{ij}^k \} }\textit{ } \forall \textit{ } k \textit{ } |\textit{ } A_k  \in \mathcal{R}, \textit{\ \ where,   } 
\\ 
\bigtriangleup_{ij}^k = | h_{iA_k} - h_{jA_k} |. 
\end{split} 
\label{deltaij}
\end{equation}
\begin{lemma}
%An edge % between nodes 
%	$(i,j)$  is invisible with respect to $ \mathcal{R} $ only if \[ {Max| h_{iA_k} - h_{jA_k}  |  \leq 1} ; \textit{ } \forall \textit{ } i,j \in \mathcal{V} \textit{ } and \textit{ } A_k  \in \mathcal{R}	\]
An edge % between nodes 
$(i,j),  \textit{ } i,j \in \mathcal{V},$ is invisible with respect to $ \mathcal{R} $ only if  $i   \notin \mathcal{R}, \textit{ } j   \notin \mathcal{R},$ and $\bigtriangleup_{ij} =1.$
% \[ {Max| h_{iA_k} - h_{jA_k}  |  \leq 1} ; \textit{ } \textit{ } A_k  \in \mathcal{R}.	\]
Conversely,  this is also a necessary condition for an edge $(i,j)$ to be ambiguous w.r.t. a given $\textbf{P}_{{\mathcal{R}}}$.
\end{lemma} 
\begin{proof}%\vspace{-3mm}
%	Anchor $k$ is the only node with the corresponding coordinate equal to zero, i.e., $h_{iA_k}=0 \  \textit{iff}\  i=A_k$. Furthermore, $h_{iA_k}=1$ if and only if there is an edge between $i$  and $A_k$.
Nodes have unique distance vectors in case of a resolution set, and thus $ \bigtriangleup_{ij} \neq 0$ for $i \neq j$. 	From Lemma \ref{lemma_anchor_edge}, for an edge to be invisible,  $i,j   \notin \mathcal{R}$.  
		Consider edge $(i,j), \textit{ }i,  j   \notin \mathcal{R},$  and  a landmark $A_k$ such that  $ \bigtriangleup_{ij}^k = | h_{iA_k} - h_{jA_k} |>1 $. Without loss of generality, let $ h_{iA_k} \geq h_{jA_k} $. If there is an edge  $(i, j)$, then there is a shorter path from $A_k$ to $i$ consisting of the shortest path from $A_k$ to $j$ followed by the edge $(i,j)$ and thus, $ h_{iA_k} = h_{jA_k} + 1  $ which contradicts the assumption. Thus, given $ \bigtriangleup_{ij}^k > 1 $, there can be no edge, invisible or otherwise, between nodes $i$ and $j$. 
%	Consider edge $(i,j), \textit{ }i,  j   \notin \mathcal{R},$  and  an anchor $A_k$ such that $ \bigtriangleup_{ij}^k > 1 $,  where $ \bigtriangleup_{ij}^k = | h_{iA_k} - h_{jA_k} | $. Without loss of generality, let $ h_{iA_k} \geq h_{jA_k} $. If there is an edge  $(i, j)$, then there is a shorter path from $A_k$ to $i$ consisting of the shortest path from $A_k$ to $j$ followed by the edge $(i,j)$ and thus, $ h_{iA_k} = h_{jA_k} + 1  $ which contradicts the assumption. Thus, given $ \bigtriangleup_{ij}^k > 1 $, there can be no edge, invisible or otherwise, between nodes $i$ and $j$.
\end{proof}
\section{Construction set and link dimension}
\label{construction set and link dimension}
Next, we introduce the novel concepts of \textit{construction set} and \textit{link dimension}   to facilitate exact representation and reconstruction of a graph via a set of  distance vectors. 
\begin{definition}
	Link dimension:
	A “construction set”, $ \mathcal{C} $, is defined as a set of landmarks that allows exact and unambiguous construction of $G$ from the formed distance vectors, $\textbf{P}_\mathcal{C}$. Such a construction set with minimum cardinality is a minimum construction set, $\mathcal{\tilde{C}}$, and its cardinality is called the “link dimension,” $ \gamma(G) $. 
% keep in thesis----------	
% All the members of set $ \{\tilde{C}\} $ can be referred to as “constructors.”
% 	accurate determination of all the edges with no missing or extra edges as compared to the original $ G $ 
\end{definition}
%\begin{corollary}
%	Given a construction set $\mathcal{C}$, $\exists$ a unique graph $G$ associated with it.
%% 	Distance vector matrix corresponding to $ \gamma(G) $ gives a unique and complete representation of $ G $.
%\end{corollary}
%\begin{figure}[!htbp] 
%	\centering 
%	\includegraphics[width=0.36\textwidth]{invisible_edges_5}	
%	\caption{
%		%In this figure we show a 
%		No invisible edges even with sufficient distance vectors
%	}
%	\label{no invi}
%\end{figure}
% keep in thesis
% \begin{definition}
% 	Adjacent nodes:
% 	Two nodes $ i $ and $ j $ are said to be \textit{"adjacent"} to each other if $ A[i,j] = 1 $ i.e. there is an $ edge(i,j) $ connecting these nodes. Otherwise they are said to be \textit{"non-adjacent"}.
% \end{definition}
% The way in which a construction set,  $ \mathcal{C} $ differs from a resolution set $ \mathcal{R} $ is that $ \mathcal{R} $
%A resolution set $ \mathcal{R} $  guarantees only the uniqueness of distance vectors for each node, whereas, $ \mathcal{C} $ provides the extra information needed %, which may not be contained in $ \mathcal{R} $, 
% to determine the presence or absence of a link between any two nodes.
By definition, there exists a unique graph $G$ associated with the distance vectors $\textbf{P}_\mathcal{C}$ of a given construction set $\mathcal{C}$. 
For example,   $\tilde{ \mathcal{C}}=\{1,2,3\}$ is a minimum construction set for each of the  graphs shown in Figure \ref{fig_1:sdim} and \ref{fig_1:another graph}, and the corresponding $\textbf{P}_{\tilde{\mathcal{C}}}$s  uniquely identify the respective graphs. In contrast,  ${\tilde{\mathcal{R}}}=\{1,2\}$ for graphs in Figure \ref{fig_1:figa} and \ref{fig_1:two figures} result in the same $\textbf{P}_\mathcal{R}$.
There may however be multiple minimum construction sets associated with   $G$. 
% ${\tilde{\mathcal{R}}}=\{1,2\}$ for graphs in Figure \ref{figa} and \ref{two figures} result in the same $\textbf{P}_\mathcal{R}$, whereas 
% %$ \mathcal{R} $ may provide us with two distinct vectors such as $\langle1,2\rangle$ and $\langle2,1\rangle$ in Figure \ref{figa} and \ref{two figures}. % however, as both the hop differences are only a single hop apart, a link may exist between these two nodes, as shown in Figure \ref{two figures}. 
%   $\tilde{ \mathcal{C}}=\{1,2,3\}$ %$ \mathcal{C} $
% provides sufficient reference distances to distinguish between adjacent nodes and non-adjacent nodes for the two graphs as shown in Figure \ref{sdim} and \ref{another graph}. % It illustrates how adding additional reference point(s) can resolve the ambiguity of structures, shown in Figure \ref{two figures}. 
 %Thus, there can be multiple realizations for distance vectors of minimum resolution set, whereas a construction set uniquely defines a graph for its distance vectors. 
% Consequently, the distance vector set for a minimum construction set gives the minimum sampling of data required for complete and unique graph reconstruction.
%Construction set, $ \gamma(G) $, preserves the adjacency of nodes in a graph.

\begin{theorem}\label{md less that ld}
	%     A ??? sufficient???   necessary  condition for a set of landmarks $\mathcal{M}$ to be a construction set,\textit{ }$\mathcal{C}$, is a) $\mathcal{M}$ is a resolution set,\textit{ }$\mathcal{R}$, and b) for each unconnected node pair $i,j$ in $G$, $\exists$ at least one anchor, $A_k$, such that $ \bigtriangleup_{ij}^k > 1 $ ( i.e., $ \bigtriangleup_{ij} > 1 $).
	A necessary  condition for a set of landmarks $\mathcal{M}$ to be a construction set,\textit{ }$\mathcal{C}$, is a) $\mathcal{M}$ is a resolution set, and b) for each unconnected node pair $i,j$ in $G$, $\exists$ at least one landmark $A_k$, such that $ \bigtriangleup_{ij}^k > 1 $ ( i.e., $ \bigtriangleup_{ij} > 1 $).
\end{theorem}
\begin{proof}
a)	When $ \mathcal{M} $ is not a resolution set, there are at least two nodes, $ i $ and $ j $, that have identical distance vectors, i.e., $ \bigtriangleup_{ij} =0 $, and the edge between the two nodes is ambiguous.
b)	When $ \bigtriangleup_{ij} = 1$ for nonadjacent nodes $i,j$, it is possible to add an edge $\langle i,j \rangle$ without changing the distance vectors resulting in a different graph where the two nodes are connected. To avoid such ambiguity, it is necessary to have $ \bigtriangleup_{ij} > 1 $.  % Furthermore, the distance vector elements are independent of the presence or absence of an edge between these two nodes. As there are two graphs satisfying the same distance vector set, $ \mathcal{M} $ is not a construction set. Thus, $ |\mathcal{M}| \neq \gamma(G) $.
\end{proof}
\begin{corollary} The link dimension $\gamma(G)$ and the metric dimension $\beta(G)$ of a  graph G are related by, 
\begin{equation} \label{link and metric relation}
\gamma(G) \geq \beta(G).        
\end{equation}
\end{corollary}
\begin{proof}
A minimum construction set $ \tilde{ \mathcal{C}} $ is a resolution set, while a minimum resolution set does not necessarily have a sufficient number of landmarks to ensure %$ \bigtriangleup_{ij} > 1 $.
exact reconstruction of $G$ (e.g., see Fig. \ref{fig_1:figa},\ref{fig_1:two figures}).
\end{proof}

%%Satisfying metric dimension gives us the minimum solution for eliminating $ (\bigtriangleup_{ij}=0) $.
%Therefore, starting from a graph $ G $ and its $ \textbf{P}_\mathcal{R} $, where $ |\mathcal{R}| \geq \beta(G) $, a construction set can be obtained as follows: 
%\begin{enumerate}
%\item Identify the pairs of non adjacent nodes $i,j \in \mathcal{V} $  for which  $ \bigtriangleup_{ij}=1.$ \vspace{-1.5mm}
%\item Select as a new landmark a node $ k \in \mathcal{V} $ s.t. it resolves the invisible edge between nodes $ i $ and $ j, $ i.e., makes $ \bigtriangleup_{ij} > 1  $. (A simple but not necessarily optimal  choice is to select either $i$ or $j$ as a landmark. A more efficient procedure may attempt to resolve multiple edges using one landmark). \vspace{-1.5mm}
%\item Repeat above steps  until no non-adjacent node pairs are left for which $ (\bigtriangleup_{ij} = 1) $.  %$ \textit{ } \forall \textit{ } a_{ij} = 0; \textit{ } \forall \textit{ } i,j \in \mathcal{V} $.  
%The resulting set of landmarks is a construction set for $ G$.
%\end{enumerate} Conversely, given $ \textbf{P}_\mathcal{C} $ for a construction set $ \mathcal{C}$, G can be constructed by assigning edges  $(i,j)$ for   all $i,j$ for which $ \bigtriangleup_{ij} =1$.
A resolution set is not a construction set when there are edges that are invisible to its landmarks.  An edge $(i,j) \in \mathcal{E}$ is invisible if its removal from $G$ does not change $ \textbf{P}_\mathcal{R}$, while an edge $(i,j) \notin \mathcal{E}$ is invisible if its addition does not change $ \textbf{P}_\mathcal{R}$. Therefore, given a resolution set $\mathcal{R}$ for a graph $G$, a construction set can be obtained as follows:
\begin{enumerate}
	\item Identify the set of edges $\mathcal{I}$ that satisfy the necessary conditions to be  invisible or ambiguous:
	\begin{equation}
	\mathcal{I}= 
	\left\{
	(i,j) \Biggl|
	\begin{split}
	(i,j) \in \mathcal{E} \textit{ and removal of } (i,j) \textit{ from } G \textit{ does not change the distance vectors of } i \textit{ and }j; \\
	\textit{Or\ } 
	(i,j) \notin \mathcal{E} \textit{ and } \bigtriangleup_{ij} = 1
	\textit{\hspace{7cm} \ \ \ \ \ \ \ \ \  }  
	\end{split}
	\right\}
	\end{equation}
%	  {$(i,j) \in \mathcal{E}$ such that its removal from $G$ does not change the distance vector values of nodes $i$ and $j$, and\\ 
%	  $(i,j) \notin \mathcal{E}$ such that $ \bigtriangleup_{ij} = 1.$ } \vspace{-1.5mm}
	\item Select additional landmarks to resolve the  edges in $\mathcal{I}$. (A simple but not necessarily optimal  choice for resolving edge $(i,j)$ is to select either $i$ or $j$ as a landmark. Resolving multiple edges using one landmark will produce a more compact solution). %\vspace{-1.5mm}
	%\item Repeat above steps  until no invisible edges are left.   
	The resulting set of landmarks is a construction set for $ G$.
\end{enumerate} 
Conversely, given $ \textbf{P}_\mathcal{C} $ for a construction set $ \mathcal{C}$, G can be constructed by assigning edges  $(i,j)$ if and only if $ \bigtriangleup_{ij} =1$.
%\section{Bounds for Link Dimension}
\section{Properties and Bounds}
\label{bounds for link dimension}
\setlength{\belowdisplayskip}{0pt} \setlength{\belowdisplayshortskip}{0pt}
\setlength{\abovedisplayskip}{0pt} \setlength{\abovedisplayshortskip}{0pt}
% describe what you do in this section
%Next we derive several   properties and bounds for $\gamma(G)$. % under different constraints.
Consider % $ m $ landmarks in 
a graph with diameter $d$, i.e. $ h_{ij} \leq d\textit{ } \forall\textit{ } i,j \textit{ } \in \textit{ } \mathcal{V}$. 
%Each distance vector has $ m $ elements representing distances from respective landmarks,  and 
No element of a DV for a non landmark node can hold values outside $1$ to $d$.
The $ m $ landmarks themselves have unique DVs with one of the elements equal to zero. Thus, $m$ landmarks yield at most $(d^m+m)$ unique coordinates.
%$(d^m+m)$ gives us the maximum number of unique distance vectors with $m$ landmarks. 
%As we need %at least $m$ nodes 
%to generate unique distance vectors for each of the  $N$ nodes,
% a bound for metric dimension is given by \cite{9}:
Therefore, 
metric dimension $(\beta(G)=m)$ is bounded by \cite{9}:
\begin{equation}
(m+d^m \geq N)
\label{metric dimension bound}
\end{equation}
%The metric dimension bound  in Eq. \ref{metric dimension bound}
Above bound deals with $N$, the number of nodes to be resolved. 
In contrast, a construction set has to resolve each of the possible links, i.e.,  
be able to distinguish among each of the $L$ edges present as well as  unambiguously exclude edges among non-adjacent nodes. %, i.e., a construction set has to  resolve each link in $G$.
%The metric dimension bound  in Eq. \ref{metric dimension bound} deals with $N$ because the landmarks help  resolve each node. 
%In contrast, a link dimension bound 
%needs to be able to distinguish among each of the $L$ links as well as exclude all the missing edges, i.e., a construction set has to  resolve each link in $G$.
To help obtain bounds for the  link dimension, we consider the relationship between DVs for an edge to be feasible.  
 \begin{definition}
	Feasible links: An edge is feasible between nodes $i$ and $j$ % only
	 if $\bigtriangleup_{ij}^k \rlap{\kern.45em$|$}> 1$  $\forall \textit{ } k$, i.e. $\bigtriangleup_{ij} \rlap{\kern.45em$|$}> 1.  $ % where, $\bigtriangleup_{ij} =  {Max \{ \bigtriangleup_{ij}^k \} },\textit{ } \forall \textit{ } k $.
	\label{adj}
\end{definition}
\noindent
Thus, conceptually, in addition to the bound in Eq. \ref{metric dimension bound}, $m = \gamma(G)$ must also satisfy
%Bound for link dimension can be conceptually given by:
\setlength{\belowdisplayskip}{0pt} \setlength{\belowdisplayshortskip}{0pt}
\setlength{\abovedisplayskip}{0pt} \setlength{\abovedisplayshortskip}{0pt}
\begin{equation}
\begin{split}
(\textit{Number of feasible links for a coordinate set with } (d^m+m) \textit{ unique vectors}) 
\\
\geq (\textit{Actual number of links in the graph}) \\
%(\textit{links in the original network with } N \textit{ nodes}) 
\end{split}
\label{link dimension bound 1}
\end{equation}                    
% \noindent A complete graph has all feasible links. 
%While $ L $: actual links in the given graph, note, that the metric dimension bound Eq. \ref{metric dimension bound} deals with $N$ because the landmarks help to resolve each node in $G$; whereas link dimension bounds 
%%Eq. \ref{link dimension bound 1} 
%will deal with $L$ as construction set resolves each link in $G$. Thus, conceptually,
%%Bound for link dimension can be conceptually given by:
%\setlength{\belowdisplayskip}{0pt} \setlength{\belowdisplayshortskip}{0pt}
%\setlength{\abovedisplayskip}{0pt} \setlength{\abovedisplayshortskip}{0pt}
%\begin{equation}
%(\textit{Feasible links in a coordinate set with } (d^m+m) \textit{ unique vectors}) \geq (\textit{Actual links in the original network}) 
%%(\textit{links in the original network with } N \textit{ nodes}) 
%\label{link dimension bound 1}
%\end{equation}                                                        
%Note, that metric dimension bound Eq. \ref{metric dimension bound} deals with $N$ because the landmarks help to resolve each node in $G$; whereas link dimension bound Eq. \ref{link dimension bound 1} deals with $L$ as construction set resolves each link in $G$.
%Let’s say,
%%$ L $: actual links in the given graph,
%$ m $: number of landmarks and
%$ N $: number of vertices in $ G $.

Next, we derive bounds for link dimension, $m = \gamma(G)$, under two different constraints:\\
\noindent \textbf{Case I}: The number of links (edges) $ L $  of $G$ is known\\
%, nodes $ N $ and $m$ anchors, each node will have a distance vector of length $m$.
%Given node $i$ with distance vector $\langle x_1,x_2,..,x_m\rangle$ and node $j$ with distance vector $\langle y_1,y_2,..,y_m\rangle$,  link $(i,j)$ is feasible only if $\bigtriangleup_{ij} = 1$. Thus, with respect to  edge $(i,j)$,  each $y_k$. $k = \{1,2..,m\}$  is limited to three  values from $\{x_k,x_k - 1, x_k + 1\}$. Hence, a given node can have $ (3^m - 1) $ possible combinations of adjacent distance vectors (neighbors/links).
%For a landmark node, however, one of the distance values is 0, and $x_i -1$ is not possible. 
Link $(i,j)$ is feasible only if $\bigtriangleup_{ij} = 1$. Thus, with respect to  edge $(i,j)$,  each $h_{jA_k}$, $k = \{1,2..,m\}$  is limited to three  values from $\{h_{iA_k},h_{iA_k} - 1, h_{iA_k} + 1\}$. Hence, a given node can have $ (3^m - 1) $ possible combinations of adjacent distance vectors corresponding to neighbors or links.
For a landmark node ($A_k$) however, the only possible distances for an adjacent node $i$ is ${h_{iA_k}}=1$.  
Thus, a  landmark can have only up to $ 3^{(m-1)} $ links. Thus, Eq. \ref{link dimension bound 1} can be restated as
% For two adjacent nodes $i$ and $j$, $ \bigtriangleup_{ij} \leq 1 $. Note, that the difference can be 0 between $(i,j)$ when distance vectors are not unique.
%Feasible neighbor links for a node with distance vector of size $m$ can be given as 〖$(3^m-1)$ for a general node and $(3^{m-1})$ for a landmark node.
\begin{equation}
〖{(3〗^m-1)×(N-m)}+{(3^{m-1})(m)}  \geq (L\times2)              
\label{feasible links 2}
\end{equation} 
% \textbf{Constraint II}: Only “$\textit{N}$ is known for the graph”\\
%Note that the value of $\textit{L}$ lies within $\textit{N}-1$ and $(\textit{N}^2-\textit{N})/2$ 
%for a connected graph. %i.e. for a line and for a complete graph respectively. 

% \begin{equation}
% 〖{(3〗^m-1)×(N-m)}+{(3^{m-1})(m)}  \geq {N^2}-N            
% \label{tighter lower bound}
% \end{equation} 
%Let’s say,
%L: actual links in the given graph,
% \setlength{\belowdisplayskip}{0pt} \setlength{\belowdisplayshortskip}{0pt}
% \setlength{\abovedisplayskip}{0pt} \setlength{\abovedisplayshortskip}{0pt}
% For $ n_{dmin} $ as the minimum node degree of the graph,
% we get minimum links in $ G $ thus,
% \begin{equation}
% %(m+d^m ) \times {\frac{n_d}{2}} \geq L
% (m+d^m ) \times {\frac{n_{dmin}}{2}} \leq L
% \label{B1}
% \end{equation}
% $ LHS \leq L $ in Eq. \ref{link dimension bound 1}.
%\noindent $ n_{dmax} $ gives maximum links in $ G $ thus,
\noindent \textbf{Case II}: “Maximum node degree $(n_{dmax}) $ and number of links ($\textit{L}$) of $G$  are known\\
%for the graph”:
Eq.\ref{metric dimension bound} and  $n_{dmax} $ yields the bound in Eq. \ref{B2}, while the number of links possible for a landmark yields Eq. \ref{max links eqn}: 
%Using Eq.\ref{metric dimension bound} %and Eq. \ref{B2} for this case, we get the following bounds that have to be satisfied:
\setlength{\belowdisplayskip}{0pt} \setlength{\belowdisplayshortskip}{0pt}
\setlength{\abovedisplayskip}{0pt} \setlength{\abovedisplayshortskip}{0pt}
\begin{equation}
%(m+d^m ) \times {\frac{n_d}{2}} \geq L
(m+d^m ) \times {\frac{n_{dmax}}{2}} \geq L
\label{B2}
\end{equation}
\setlength{\belowdisplayskip}{-1pt} \setlength{\belowdisplayshortskip}{-1pt}
\setlength{\abovedisplayskip}{-2pt} \setlength{\abovedisplayshortskip}{-2pt}
% $ LHS \geq L $ in Eq. \ref{link dimension bound 1}.
% \noindent If $n_d = n_{dmax} = n_{dmin} $ for the given graph, 
% then, from Eq.s \ref{metric dimension bound} and \ref{link dimension bound 1},
%\setlength{\belowdisplayskip}{0pt} \setlength{\belowdisplayshortskip}{0pt}
%\setlength{\abovedisplayskip}{0pt} \setlength{\abovedisplayshortskip}{0pt}
%\noindent Thus, from Eq. \ref{B2}, we can say that
\begin{equation}
(3^{(m-1)}) \geq n_{dmax}
\label{max links eqn}
\end{equation}
\begin{proposition} %\begin{lemma}
    Linear graph  (path) is the only graph with $\gamma(G) = 1$.
\end{proposition} %{}lemma}
\begin{proof}%\vspace{-3mm}
If %a construction set 
$ \mathcal{C} $ has only one landmark, and $G$ has $N$ nodes, the only possible set of DVs is $ \{ \langle 0 \rangle, \langle 1 \rangle,.. \langle N-1 \rangle \} $.  The landmark has a coordinate $0$, and it is connected to exactly one node (which has coordinate $1$),  which in turn is also connected to the node with coordinate $2$, etc.  This is the line graph, with the anchor placed at one end. 
%    For a line graph, the given one anchor at either end of the line, produce distance vectors $ \{ \langle 0 \rangle, \langle 1 \rangle,.. \langle N \rangle \} $. For such distance vectors, the only graph possible is a line graph with $N$ nodes.
    % For a line graph, the given one anchor at either end of the line, produce $\bigtriangleup_{ij} = 1 \iff a_{ij} = 1$ resolving all the edges.
\end{proof}
\begin{proposition} %\begin{proposition}
   For a complete graph with $N$ nodes, link dimension $\gamma(G) = N-1$.
\end{proposition} 
\begin{proof} %\vspace{-3mm}
	If two of the nodes are not in ${\mathcal{C}}$, both of them have identical distance vectors, consisting of all 1's. With $N-1$ nodes in ${\mathcal{C}}$, each edge is connected to a landmark node and is  thus resolved. 
%    Any two nodes in a complete graph have exact same path lengths to all anchors creating identical distance vectors. Thus, metric dimension for complete graphs is $N-1$ which resolves all the edges in $G$. \todo{proof incorrect}
\end{proof}
%%%%%%%%%%%%%%%%%%%%%%INCLUDE IF SPACE
\begin{proposition}
   For a finite cyclic graph $C_N$ of  $N > 6$ nodes, the link dimension $\gamma(G)$ is $ 2$.
\label{prop_cyclic}
\end{proposition}
\begin{proof}  % \vspace{-3mm}
	    %Thus, all nodes have unique distance vectors with $|\mathcal{M}| = 2$. 
	    Let the $N$ nodes be  labeled $[1,2,..,N]$. $\mathcal{C} = \{1,3\}$ is a minimum construction set for this graph. %As $\beta(G) = 2$ for $C_N$ ( $N > 2)$ \cite{9}, $\{1,3\}$ is a minimum construction set.
%    $\beta(G) = 2$ for $C_N$ and $N > 2$ \cite{9}. Thus, all nodes have unique distance vectors with $|\mathcal{M}| = 2$. Consider a cyclic graph with $N$ nodes with labels $[1,2,..,N]$ and $\mathcal{M} = \{1\}$. Thus, each non-landmark node creates an ambiguous edge with two to three other nodes. As we add another landmark at two hops from first landmark, s.t. $\mathcal{M} = \{1,3\}$, each node pair that had an ambiguous edge is now separated by at least two hops and at the maximum of $d$ hops, $d$ being the diameter of $C_N$.
    % Ambiguous edges, if any, would exist between nodes $i,j$ with distance vectors $\langle x,y \rangle$ and $\langle x-1,y+1 \rangle$ to give $\bigtriangleup_{ij} = 1$. When third landmark is added from the neighborhood of current landmarks, distance vectors of nodes $i,j$ would change to $\langle x,y,z \rangle$ and $\langle x-1,y+1,z+2 \rangle$. Thus, the third coordinate would create a difference of at least two units, consequently, resolving the ambiguous edges. 
\end{proof}

%\textcolor{red}{-should we make this as a separate section?? - Relationship between link dimension and metric dimension}\\
%Sebo and Tannier \cite{6} introduced the concept of {\it strong metric dimension}. 
%A concept closely related  to metric and link dimensions is  that of {\it strong metric dimension}, $sdim(G)$ \cite{6}.
Closely related   to metric dimensions is the {\it strong metric dimension}, $sdim(G)$, the minimum cardinality of a strong resolution set \cite{11}, \cite{6}.
Node $w$ strongly resolves two nodes $ u $ and $ v $ if $ u $ belongs to a shortest  $ v-w $ path or if $ v $ belongs to a shortest  $ u-w $ path. The node set $ \mathcal{S} $ %of $ G $
 is a strong resolving set %of $ G $ 
if every two distinct nodes of $ G $ are strongly resolved by some vertex in $ \mathcal{S} $. 
% If $\mathcal{S}$ is a minimum strong resolution set, then strong metric dimension, $sdim(G) = |\mathcal{S}|$ \cite{11}.
%$sdim(G)$, %of graph $ G $  is the minimum cardinality of the strong resolution set \cite{11}. 
If, as has been observed in \cite{6,12},  a strong resolution set can uniquely determine a graph, then  $ \mathcal{S} $ is a construction set. The question then is whether it is always a minimum construction set. Cyclic graph $C_N, N>6$ shows it is not. According to  Proposition \ref{prop_cyclic}, $\gamma({C_N})= 2$, which is less than  $sdim(C_N) = N/2$ \cite{strong_ex}. Therefore, 
\begin{equation} 
 sdim(G) \geq \gamma(G) . 
\end{equation}  

\section{Conclusion}
\label{conclusion}
A construction set of a graph $G$ is a set of landmark nodes such that the set of distance vectors of the nodes  allows the exact construction of $G$. The minimum number of landmarks in a construction set is the link dimension of $G$. Similar to a  resolution set it preserves unique labeling of nodes, but uses additional landmarks as necessary to resolve all the edges. The corresponding set of distance vectors therefore provides an exact representation of $G$ itself, thus  extending the distance vector based approach beyond unique labeling of nodes %, which was the case with resolution sets,
to unique identification of graphs. 
\vspace{-0.3in}
% \\{\bf References}
\begin{singlespace}
\bibliography{mybibfile}
\end{singlespace}	
\end{document}